\documentclass[a4paper,11pt]{article}
\usepackage{hyperref}
\usepackage{url,graphicx}
\usepackage{amsthm,amsmath,amssymb,amsfonts}
\usepackage{showexpl}
\usepackage{mathrsfs,xr}
\usepackage[all,arc,curve,color,frame]{xy}
\usepackage{graphics,fullpage}
\xyoption{all}

\title{QUANTIZATION OF THE  $G$-CONNECTIONS
VIA THE TANGENT GROUPOID}

\author{Alan Lai
\thanks{Email: alan@caltech.edu -- The research is 
supported under the sponsorship of The Croucher Foundation}
\\
California Institute of Technology
}

\begin{document}
\maketitle

\newtheorem{definition}{Definition}[section]
\newtheorem{example}{Example}[section]
\newtheorem{theorem}{Theorem}[section]
\newtheorem{lemma}[theorem]{Lemma}
\newtheorem{corollary}[theorem]{Corollary}
\newtheorem{proposition}[theorem]{Proposition}
\newtheorem{remark}[definition]{Remark}

\newcommand{\su}{\mathfrak{su}}
\newcommand{\SO}{\operatorname{s}}
\newcommand{\RA}{\operatorname{r}}
\newcommand{\Gg}{\mathcal{G}}
\newcommand{\Tt}{\mathcal{T}}
\newcommand{\Ll}{\mathcal{L}}
\newcommand{\Aa}{\mathcal{A}}
\newcommand{\g}{\mathfrak{g}} 
\newcommand{\KER}{\operatorname{ker}}
\newcommand{\DIFF}{\operatorname{Diff}}
\newcommand{\ENDO}{\operatorname{End}}
\newcommand{\DOM}{\operatorname{Dom}}
\newcommand{\Cc}{\mathbb{C}}
\newcommand{\TR}{\operatorname{Tr}}
\newcommand{\DET}{\operatorname{Det}}
\newcommand{\Tr}{\operatorname{tr}}
\newcommand{\Th}{\operatorname{th}}
\newcommand{\SYM}{\operatorname{Sym}}

\setcounter{section}{-1}

\begin{abstract}
A description of the space of $G$-connections using the tangent groupoid is given.
As the tangent groupoid parameter is away from zero, the $G$-connections act as convolution operators
on a Hilbert space.
The gauge action is examined in the tangent groupoid description of the $G$-connections.
Tetrads are formulated as Dirac type operators. The connection variables and tetrad variables in Ashtekar's gravity are presented as operators on a Hilbert space.
\end{abstract}

\setcounter{tocdepth}{2}


\section{Introduction}
In Ashtekar's theory of gravity, a $SU(2)$-connection captures the extrinsic curvature of a space-like leaf
in a time-transversal foliation, and the intrinsic geometry of the leaf is given by a tetrad \cite{ashtekar}.
It is shown that the Einstein-Hilbert functional and Einstein equations can be written
in terms of  $SU(2)$-connections and tetrads, thus
these variables together  recast Einstein's theory of gravity.
By rewriting Einstein's gravity with the connection variables and tetrad variables, 
one could attempt to quantize gravity via the Hamiltonian formalism, and obtain a theory of quantum gravity \cite{ashham}.
The reader may refer to Thiemann's introductory \cite{LQG}.
This article offers an alternative view of the space of connections to ones that appear in other loop quantum gravity literatures.

To take the connections as dynamic variables in a quantum theory, one studies wave functions on the space of connections. Unfortunately, the lack of a measure on such a space already poses the first challenge.
One typical solution to obtain a measure on the connection space comes from spaces of progressively 
refined cylindrical functions. In another description, one uses  a finite set of curves to probe the space of 
$G$-connections to obtain a finite dimensional manifold  that depend on the sets of curves. By successively refining 
the finite sets of curves, one obtains a pro-manifold that extends
the  original space of connections to the so-called space of generalized connections \cite{baez}.

As a step to quantizing gravity in the Ashtekar framework, there are recent developments of describing such an extended space of connections using a spectral triple in noncommutative geometry \cite{AGN,lai}, which captures the geometry
of the space of generalized connections as operators on a Hilbert space.
While the geometries of the base manifold and the  space of $G$-connections  on it are in theory retained, the construction of  the spectral triple is considered too discrete to practically allow one to reconstruct the
base manifold and its $G$-connections.
To rid the discrete description of using finite sets of embedded curves, this article proposes an alternative way to smoothly probe the space of connections using the tangent groupoid.
The purposes of this article are to present an alternative idea to the studies of loop quantum gravity
and to generate discussions on it, therefore the content presented here will be incomplete.

This article contains four sections, they are arranged as follows:
Section 1 introduces a smooth groupoid that we use to probe the space of connections, the
tangent groupoid. We will define our q-connections as $G$-valued functions on a tangent groupoid.
In Section 2, we discuss the gauge action on the space of q-connections, and in Section 3, we argue that the connection variables and tetrad variables are quantized as operators, and present a gauge invariant function that mimics the loop variables in loop quantum gravity.

The last section is an outlook that summarizes the article, and suggests further developments in this direction.


\section{Tangent Groupoid}
In loop quantum gravity, one studies the space of connections via cylindrical functions.
These functions are obtained from probing the space of connections with  finite sets of embedded curves.
For instance, the edges of a triangulation of a compact manifold gives a finite set of embedded curves.
Given a $G$-connection, one could compute the holonomy along a given curve. 
And by assuming the $G$-bundle over the manifold is trivialized,  one obtains a  $G$-value from it.
When one evaluates the holonomy along that fixed curve of the entire space of $G$-connections, one
obtains the full group $G$. 
The mental picture is that we use an one-dimensional object to probe the connection space, and obtain a copy of $G$. Thus, there is a copy of $G$ associated to  each embedded curve.
As one only considers a finite set of curves, this coarse-grain
approximation of the connection space via holonomies is given by a finite product of $G$, which is a
finite dimensional manifold with a unique invariant measure.
By successively refining the set of finite curves, say by barycentric sub-dividing the triangulation, one
can approximate the connection space arbitrarily well, and this procedure results in a pro-manifold that
the connection space densely embeds into.
However, the finite sets of curves along with their refinements are too discrete to allow one to 
reconstruct the smooth structure of the original connection space.
In this section, we propose using a smooth groupoid to probe the connect space with, and the holonomies will be encoded by the smooth $G$-valued functions over the smooth groupoid.
\begin{definition}
\label{liegroupoid}
A {\bf Lie groupoid} is a groupoid
$\Gg  \rightrightarrows ^{\hspace{-0.25cm}^{\RA}}
_{\hspace{-0.25cm}_{\SO}}X$ with manifold structures on $\Gg$ and $X$ such that $\SO, \RA$ are 
submersions, the inclusion of $X$ in $\Gg$ as identity homophisms and the composition $\Gg \times \Gg \to \Gg$ are smooth.
\end{definition}

\begin{example}[\cite{tangent}]
\label{liegroupoidex}\mbox{ }
\begin{enumerate}
\item
The tangent bundle $T\Sigma$ of a manifold $\Sigma$ forms a Lie groupoid $T\Sigma  \rightrightarrows ^{\hspace{-0.25cm}^{\RA}}
_{\hspace{-0.25cm}_{\SO}}\Sigma$ with the \emph{source} and \emph{range} maps 
$\SO, \RA: T\Sigma\to \Sigma$ given by 
$\SO(x,V_x) = x = \RA(V_x,x)$ for $(x,V_x) \in T\Sigma$, the inclusion $\Sigma \hookrightarrow T\Sigma$ given by
the zero section, and composition $T\Sigma\times T\Sigma \to T\Sigma$ given by $(x,V_x )\times (x,W_x) \mapsto 
(x,V_x + W_x)$.
\item
The product $\Sigma\times \Sigma$ forms a Lie groupoid $\Sigma \times \Sigma
  \rightrightarrows ^{\hspace{-0.25cm}^{\RA}}
_{\hspace{-0.25cm}_{\SO}} \Sigma $ with the \emph{source} and \emph{range} maps 
$\SO, \RA: \Sigma \times \Sigma \to \Sigma$ given by 
$\SO(x,y ) = x $, $ \RA(x,y)=y$, the inclusion $\Sigma \hookrightarrow \Sigma\times \Sigma$ is given by the 
diagonal embedding, and the composition 
$(\Sigma\times \Sigma) \times (\Sigma\times \Sigma) \to \Sigma \times \Sigma$ is given by
$(x,y) \times (y,z) \mapsto (x,z)$.

\end{enumerate}

\end{example}
These two examples are the only Lie groupoids this article will deal with. 
Now, let $\Sigma$ be a compact oriented manifold.

\begin{definition}[\cite{ncg}]
\label{tangentgroupoid}
The {\bf tangent groupoid} $\Tt \Sigma$ of a manifold $\Sigma$ is 
$\Tt \Sigma= T\Sigma \times \{ 0 \} \bigsqcup \Sigma \times \Sigma \times
(0,1]$ as a set.
And the groupoids 
 $T\Sigma $ and $ \Sigma \times \Sigma $ are glued together such that
for
$
p(0) =  
(x,V_x,0) \in T\Sigma \times \{0\}$,
then $p(\hbar)= \left(x,\exp \hbar V_x, \hbar\right)\in \Sigma \times \Sigma \times (0,1]$.
\end{definition}

We think of an element $(x,y,\hbar)$ of $\Sigma\times \Sigma \times (0,1]$ as
a geodesic starting at $x$ and ending at $y$, and $\hbar$ is the time it takes to travel from $x$ to $y$ 
in a given velocity. Hence, it is considered a one dimensional object.

Let $G$ denote a compact Lie group and $\g$ its Lie algebra.
Denote by $\Tt _e G $ the set of elements of $\Tt G$ with source $e$.
$\Tt_e G$ is nothing but $\g\times \{0\} \bigsqcup G\times (0,1]$ as a set, and $\g $ and $G$ are glued together by the exponential map.
Denote by $\Aa_\hbar$ the space of smooth functions from $\Sigma \times \Sigma$ to $G$ and 
$\Aa_0$ the space of smooth functions from $T\Sigma$ to $\g$ such that $A_0(x,\cdot):
T_x \Sigma \to \g$ is linear for $A_0 \in \Aa_0$ and each $x\in \Sigma $.\\

Here we think of an element $A_\hbar$ of $\Aa_\hbar$ as a holonomy presentation of a connection
for paths described by $\Sigma\times \Sigma \times (0,1]$.

\begin{definition}
\label{qconnection}
Define the space of {\bf q-connections} $F(\Tt \Sigma, \Tt_e G)$ to be
$\Aa_0 \times \{0\} \bigsqcup \Aa_\hbar \times (0,1]$ as a set.
And $\Aa_0$ and $\Aa_\hbar$ are glued together as
$A_\hbar (x, \exp \hbar V_x ) = \exp (\hbar A_0(x,V_x))$ for all $\hbar \in [0,1]$.
\end{definition}

\begin{remark}
The glueing condition of $F(\Tt \Sigma, \Tt_e G)$ implies that $A_h (x,y) \cdot A_h (y,x) \longrightarrow 
^{\hspace{-0.6cm}\hbar \to 0 } I_G$ and $A_h(x,x) \longrightarrow 
^{\hspace{-0.6cm}\hbar \to 0 } I_G$, where $I_G$ is the identity of $G$.
\end{remark}

\begin{remark}
Each $A_0\in \Aa_0$ gives rise to a $\g$-valued 1-form in a unique way. Hence,
$\Aa_0$ is naturally identified as the space of $G$-connections.
\end{remark}

$\Aa_\hbar$ forms a group under point-wise multiplication of $G$, and 
$\Aa_0$ forms a group under point-wise addition of $\g$.

\begin{proposition}
The product $F(\Tt \Sigma, \Tt_e G)$  inherits  from $\Aa_0$ and $\Aa_\hbar$ is smooth.
\end{proposition}
\begin{proof}
The proof follows from
$\left. \frac{d}{d\hbar}(A_\hbar \cdot A'_\hbar) ( x, \exp \hbar V_x) \right\rvert_{\hbar=0}  =  (A_0+A'_0 )(x,V_x) .$
\end{proof}


\section{Gauge Action}
The q-connection space $F(\Tt \Sigma, \Tt_e G)$ is a package that
captures information about probing the $G$-connection space with the tangent groupoid.
That is, an element $A_\hbar$ of $\Aa_\hbar$ evaluated at $(x,y)\in \Sigma \times \Sigma$
gives the holonomy of a connection along the geodesic from $x$ to $y$.
This holonomy presentation has a natural gauge action given by
applying a symmetry at the starting point $x$, then evaluate the holonomy along the geodesic 
from $x$ to $y$, and finally apply a reverse symmetry at $y$.
We make it formal by the following definition.

Denote by  $C^\infty(\Sigma,G)$ the set of smooth functions from $\Sigma$ to $G$.
Define the {\bf gauge action} of $C^\infty(\Sigma,G)$  on $\Aa_\hbar$ by
\begin{equation}
\label{eqn:qgaugeact}
( g\cdot A_\hbar)(x,y) := g(x) A_\hbar(x,y) g^{-1}(y)
\end{equation}
for $g\in C^\infty(\Sigma,G)$ and $A_\hbar \in \Aa_\hbar$.
And define the gauge action of  $C^\infty(\Sigma,G)$ on $\Aa_0$ by
\begin{equation}
\label{eqn:cgaugeact}
(g\cdot A_0)(x,V_x):= g(x) A_0(x,V_x) g^{-1}(x) + g(x) \langle dg^{-1}(x),V_x\rangle
\end{equation}
for $g\in C^\infty(\Sigma,G)$ and $A_0 \in \Aa_0$.
\begin{remark}
Equation~\eqref{eqn:cgaugeact} is the usual gauge action on $G$-connections.
\end{remark}

The following proposition shows that the gauge actions defined on $\Aa_\hbar $ and $\Aa_0$ are compatible.
\begin{proposition}
The  
$C^\infty(\Sigma,G)$ action  on 
the q-connection space $F(\Tt \Sigma, \Tt_e G)$ induced from Equations~\eqref{eqn:qgaugeact},\eqref{eqn:cgaugeact} is smooth.
\end{proposition}
\begin{proof}
Suppose that $A_\hbar (x,\exp \hbar V_x) = \exp \hbar A_0 (x,V_x)$.
Then 
\begin{eqnarray*}
\left. \frac{d}{d\hbar} (g\cdot A_\hbar)(x,\exp\hbar V_x) \right\rvert _{\hbar=0}
&=&
 \left. \frac{d}{d\hbar} 
 \left( g(x) A_\hbar (x,\exp \hbar V_x) g^{-1} (\exp \hbar V_x)\right) \right\rvert _{\hbar=0}\\
 &=&
 \left. g(x)\left( \frac{d }{d\hbar} A_\hbar (x,\exp \hbar V_x) \right)g^{-1} (\exp \hbar V_x) \right\rvert _{\hbar=0}
 \\&&+
 \left. g(x) A_\hbar(x, \exp \hbar V_x )\left( \frac{d }{d\hbar}g^{-1}(\exp \hbar V_x)  \right)\right\rvert_{\hbar=0}\\
 &=&
 g(x) A_0 (x,V_x)g^{-1}(x) + g(x)\langle d g^{-1} (x), V_x\rangle\\
 &=& (g\cdot A_0 )(x, V_x) .
\end{eqnarray*}
The proof is complete.
\end{proof}

Let $\DIFF(\Sigma)$ denote the diffeomorphism group of $\Sigma$.
 $\Tt \Sigma=T\Sigma \times \{ 0 \} \bigsqcup \Sigma \times \Sigma \times
(0,1]$ carries a smooth $\DIFF(\Sigma)$ action given by
\[
\begin{array}{rl}
(x,y)\mapsto \left(\sigma(x),\sigma(y) \right) & \mbox{ for } (x,y) \in \Sigma \times \Sigma 
\mbox{ and } \sigma \in \DIFF(\Sigma) ,\\
(x,V_x\mapsto  (\sigma(x), d\sigma_{\sigma(x)} V_x ) & \mbox{ for } (x,V_x) \in T\Sigma 
\mbox{ and } \sigma \in \DIFF(\Sigma) .
\end{array}
\]
Subsequently, $\DIFF(\Sigma)$ acts on the q-connection space $F(\Tt \Sigma, \Tt_e G)$ smoothly via the induced action
\[
(\sigma \cdot A_\hbar ) (x,y) = A_\hbar  \left(\sigma(x), \sigma(y) \right)
\mbox{ and }
(\sigma \cdot A_0 )(x, V_x) = A_0 \left(\sigma(x), d\sigma_{\sigma(x)} V_x\right),
\]
for $\sigma \in \DIFF(\Sigma)$.

\section{Loop Quantum Gravity}
Suppose that $G $ unitarily represents on some finite dimensional vector space, say without loss
of generality $\Cc^N$. Then $G$
 is included in the  matrix algebra $M_N(\Cc)$ as unitary matrices.
 Let us fix  a Riemannian metric $q$ on $\Sigma$.
 Then one obtains the Hilbert space 
 $L^2(\Sigma, \Cc^N)$ that $\Aa_\hbar$ acts on by convolution
 \[
 (A_\hbar * \varphi)(y)= \int _\Sigma A_\hbar (x,y) \cdot \varphi (x) dx ,
 \]
 where $A_\hbar \in \Aa_\hbar$, $\varphi \in L^2(\Sigma,\Cc^N)$, the dot $\cdot$ is the
 unitary representation of $G$, and the integration is with respect to
 the volume form induced from the metric $q$.
 
Denote by ${\Aa_\hbar^\#}$ the space of $M_N(\Cc)$-valued smooth functions on $\Sigma \times \Sigma$, thus 
${\Aa_\hbar^\#} \supset \Aa_\hbar$ and it acts on $L^2(\Sigma,\Cc^N)$.
Elements of ${\Aa_\hbar^\#}$ will again be denoted by $A_\hbar$.
${\Aa_\hbar^\#}$ comes equipped with an involution  given by the point-wise conjugate transpose of $M_N(\Cc)$.
The action of ${\Aa_\hbar^\#}$ on $L^2(\Sigma,\Cc^N)$ gives rise to 
a noncommutative product on
${\Aa_\hbar^\#} $
given by the convolution 
\[
(A_\hbar * A'_\hbar)(x, z)= \int _\Sigma A_\hbar ( x,y) \cdot A'_\hbar (y,z) dy ,
\]
for $A_\hbar, A'_\hbar \in {\Aa_\hbar^\#}$, and the integration is with respect to the volume form of $q$.

The space ${\Aa_\hbar^\#}$ forms a $*$-algebra, 
and it is identified with the ideal of trace-class operators 
 on the Hilbert space $L^2(\Sigma, \Cc^N)$.

Let $\TR$ denote the operator trace on $L^2(\Sigma,\Cc^N)$,
 and $\TR:{\Aa_\hbar^\#} \to \Cc$ is explicitly given by
\[
\TR (A_\hbar) = \int _M \Tr A_\hbar (x,x) dx ,
\]
where $\Tr$ is the matrix trace of $M_N(\Cc)$.

The linear functional $\TR: {\Aa_\hbar^\#} \to \Cc$ is invariant under the gauge action of $C^\infty(\Sigma,G)$,
as
\begin{eqnarray*}
\TR ( g \cdot A_\hbar) &=& \int _M \Tr \left( g(x) A_\hbar(x,x) g^{-1}(x) \right) dx \\
&=&  \int _M \Tr \left( A_\hbar(x,x) \right) dx \\
&=& \TR(A_\hbar) .
\end{eqnarray*}
Such a property is called {\bf gauge invariant}. 
The group element $A_\hbar(x,x)$ represents the holonomy of a connection around a loop with
base point $x$. The functional $\TR:{\Aa_\hbar^\#} \to \Cc$ being gauge invariant
is parallel to the fact that loop variables being gauge invariant in loop quantum gravity.

Let us move our attention to the metric $q$ now.
The Riemannian metric  tensor $q$ on $\Sigma$ is positive definite matrix at each point $x\in \Sigma$, it can be decomposed  as $q(x)=e(x) e^*(x)$ at each $x$ for some matrix valued smooth function
$e$ over $\Sigma$. The set of column vectors $\{e_i(x)\}$ of $e(x)$ can be thought of as an 
eigenbasis, it
 forms a frame at each $x$, and each $e_i$ is a tangent vector field of $\Sigma$. These fields of vectors are called a {\bf tetrad}.
We call the tetrad {\bf oriented} if  the point-wise determinant $\det{e(x)} >0$ for all 
$x\in \Sigma$.
A tetrad defines a metric tensor uniquely by $e e^*$, however a metric $q$ tensor could have many different decompositions $q=e e^*$.
The metric given by an oriented tetrad exhibits an $SO(d)$ symmetry, 
where $d$ is the dimension of $\Sigma$.
To wit, suppose $e'  = e T $, where $T$ is a smooth field of orientation preserving rotational matrices $SO(d)$ over $\Sigma$. Then 
\[
e' {e'} ^{*} = ( e T ) (T^* e^* ) = ee^* =q ,
\]
as $T T^* = I$.

Suppose now that $\Sigma$ is 3-dimensional and
the Lie group $G$ is $SU(2)$, so
 $\Sigma$ is a spin manifold, the symmetry group of the tetrads can be lifted from $SO(3)$ to $SU(2)=Spin(3)$.
The Lie algebra $\g=\su(2)$ of $SU(2)$ is generated by 
\[
u_1= \sqrt{\frac{-1}{2}}
\left(
\begin{array}{cc}
0& 1 \\
1 & 0 
\end{array}
\right),
u_2=
 \sqrt{\frac{-1}{2}}
\left(
\begin{array}{cc}
0& \sqrt{ -1} \\
-\sqrt{-1} & 0 
\end{array}
\right), \mbox{ and }
u_3=
 \sqrt{\frac{-1}{2}}
\left(
\begin{array}{cc}
1 &0 \\
0& -1
\end{array}
\right).
\]
The elements $\{u_i\}$ satisfy the relation  $u_i u_j + u_j u_i = -\delta_{ij}$.

Define the operator $D_q$ acting on $L^2(\Sigma, \Cc^2)$ by 
\[
D_q= u_1 e_1 + u_2 e_2 + u_3 e_3,
\]
where $q$ is a Riemannian metric on $\Sigma$ with decomposition $q=e e^*$, and
each $e_i$ is a column vectors of $e$.
The operator $D_q$  depends on the metric $q$, and is a Dirac type operator.
Notice that there is an ordering used in the definition.

The interaction between a connection and a tetrad is formulated as the commutator
$[D_q, A_\hbar]$ for $A_\hbar \in \Aa_\hbar ^\#$.

Let us compute the commutator explicitly on a Hilbert space vector $\phi \in L^2(\Sigma,\Cc^2)$:
\begin{eqnarray}
[D_q,A_\hbar ] \varphi (y)& =& 
\sum_{i=1}^3  \left( \int_\Sigma u_i \langle d_y A_\hbar (x,y),e_i(y) \rangle \varphi(x) dx
-
\int_\Sigma A_\hbar(x,y) u_i \langle d_x \varphi(x),e_i(x)\rangle dx
\right)\\
&=& \label{eqn:four}
\sum_{i=1}^3  \left( \int_\Sigma u_i \langle d_y A_\hbar (x,y),e_i(y) \rangle \varphi(x) dx
+
\int_\Sigma \langle d_x A_\hbar(x,y) , e_i(x) \rangle u_i \varphi(x) dx
\right)\\
&=& \label{eqn:five}
\sum_{i=1}^3  \int_\Sigma 
\biggl( u_i \langle d_y A_\hbar (x,y),e_i(y) \rangle 
+ \langle d_x A_\hbar(x,y) , e_i(x) \rangle u_i \biggr) 
\varphi(x) dx
\end{eqnarray}
Line~\eqref{eqn:four} follows from the Leibniz property and the boundaryless condition of $\Sigma$.

As  $\Aa_\hbar ^\#$ is a set of smooth functions, its elements preserve the domain of $D_q$,
and  the integrand\\ $\Bigl( u_i \langle d_y A_\hbar (x,y),e_i(y) \rangle 
+ \langle d_x A_\hbar(x,y) , e_i(x) \rangle u_i  \Bigr)$
in Line~\eqref{eqn:five} is again a smooth integral kernel.
Thus,
$[D_q. \cdot]$ maps $\Aa_\hbar ^\#$ back to $\Aa_\hbar ^\#$. In particular, the commutator is bounded,
hence
$\left(\Aa_\hbar ^\#, L^2(\Sigma, \Cc^2), D_q\right)$ forms a spectral triple.

\section{Classical Limit}
In the previous section, we have described the $G$-connections as integral kernels, and
the tetrads as  Dirac type operators. We also briefly examined the interaction of these two kinds
of operators. Although the procedure of turning the connections and tetrads to operators are natural,
to claim that such procedure is in fact a quantization, one needs to show that their classical limit
of their quantum interaction, the commutator, gives back their classical interaction, the Poisson bracket.
By classical limit, we mean the parameter $\hbar$ is taken to zero.

By construction, if we take $\hbar$ to zero of an element $A_\hbar \in \Aa_\hbar$, 
one obtains a map $A_0 : T\Sigma \to \g$ that is linear for each $x\in \Sigma$.
It corresponds uniquely to a Lie algebra valued one form, hence a classical connection.
The way one obtains a classical quantity out of $D_q$ is first to multiply a  scaling $\hbar$ to  $D_q$,
and represent $\hbar D_q$ in terms of the integral kernel
\begin{eqnarray}
\label{eqn:intker}
k_\hbar (y,x)=\int _{T^*_x \Sigma} \hbar p_x(z) e^{2\pi i \langle y-x , z  \rangle}  dz ,
\end{eqnarray} where $p_x(z)$ is the symbol of $D_q$. 

Since the symbol $p_x(z)$ is degree one homogeneous, $\hbar p_x(z)$ equals to 
$p_x(\hbar z)$.
It is known from symbol calculus that the first order approximation of $k_\hbar (y,x)$ is given by the symbol
$p_x(z)$, which is a function from $T^*\Sigma$ to $\su(2)$.

We have obtained the follow heuristic picture:
\[
\begin{tabular}{rc|c}
\mbox{ } & \mbox{ Connections } & \mbox{ Metrics } \\ \cline{2-3} &&\\
\mbox{ Quantum: } & $A_\hbar:\Sigma \times \Sigma \to SU(2)$  & $\hbar D$ \\
&&\\
\mbox{ Classical: } & $A_0:T\Sigma \to \su(2) $ & $p: T^*\Sigma \to \su(2)$ \\
\end{tabular} 
\]
%

The symbol of $D_q$ is an $\su(2)$-valued smooth function on $T^*\Sigma$.
 $\Aa_0$ is the space of functions from the tangent bundle $T\Sigma $ to
the Lie algebra $\su(2)$. 
By composing elements of $\Aa_0$ to the inverse metric $q^{-1}$, we
end up with $\su(2)$-valued functions on $T^*\Sigma$ again.
Unfortunately, this formalism encounters a difficulty here, because we do not have a Poisson 
algebra structure on the space of $\su(2)$-valued functions on $T^*\Sigma$.
We are hoping that there is a parent Poisson algebra that contains the space of $\su(2)$-valued functions
on $T^*\Sigma$ here, such that it allows us to give a classical limit rigourously.

We end this collection of thoughts by remarking that
 ${\Aa_\hbar^\#}$ is a set of smooth functions, it preserves the domain 
of $D_q$, and the commutator $[D_q, \cdot ]$ maps ${\Aa_\hbar^\#}$ back to ${\Aa_\hbar^\#}$.
As such, the functional $\TR:{\Aa_\hbar^\#} \to \Cc$ is invariant under the action $[D_q,\cdot]: {\Aa_\hbar^\#} \to {\Aa_\hbar^\#}$.

\section{Outlook}
In Ashtekar's gravity theory, the fundamental variables are the connections and tetrads.
We turn connection 1-forms to linear maps from the tangent bundle to the Lie algebra, then exponentiate
them to integral kernels acting on a Hilbert space.
Also we build a   Dirac type operator acting on the same Hilbert space from tetrads.
Both of these procedures of turning the connections and tetrads to operators are natural;
the original connection space can be retrieved by taking $\hbar$ to zero, while the Dirac type operator
is known to capture the Riemannian metric from noncommutative geometry techniques.
However, to truly claim the procedure of turning those variables to operators to be a quantization, one needs a Poisson bracket that corresponds to the commutator $[\hbar D_q,A_\hbar]$ in a classical sense.
While $\lim _{\hbar \to 0} A_\hbar$ gives a classical connection by construction, and
 $\lim _{\hbar \to 0} \hbar D_q$
gives the symbol of $D_q$ via oscillatory integral. Both of these can be put into the space
$C^\infty(T^*\Sigma, \su(2))$. But unfortunately, $C^\infty(T^*\Sigma, \su(2))$ is not a Poisson algebra,
because of that we do not have a genuine classical limit that corresponds to our operatorial setup.

As the metric $q$ is used to construct the Dirac type operator and the Hilbert space, it seems natural to describe metrics as unbounded Fredholm modules over our algebra $\Aa_\hbar ^\#$. 
However, the product on $\Aa_\hbar ^\#$ (the convolution) is given by the composition of the convolution action on the
Hilbert space. Therefore, before a background metric is fixed, $\Aa_\hbar ^\#$ does not incorporate
an interesting product.
It is in fact more natural to claim that the Fredholm module is over the algebra generated by the gauge group,
$C^\infty(\Sigma, M_N(\Cc))$, so one gets a spectral triple 
$(C^\infty(\Sigma, M_N(\Cc)),L^2(\Sigma,\Cc^N),D_q)$; and the space of smooth  connections is identified as the space of trace
class operators.
At this stage, the setting appears very similar to the standard model of Chamseddine-Connes-Marcolli 
\cite{stdtriple}.

  

\begin{thebibliography}{1}
  
  \bibitem{AGN}
  J. Aastrup, J. M. Grimstrup and R. Nest, 
{\em On Spectral Triples in Quantum Gravity II}, J. Noncommut. Geom. 3 (2009) 47.

\bibitem{ashtekar}
A. Ashtekar, 
{\em New variables for classical and quantum gravity}, Phys. Rev. Lett., 57 (1986), 2244-2247.

\bibitem{ashham}
A. Ashtekar, 
{\em New Hamiltonian formulation of general relativity}, Phys. Rev. D36 (1987), 1587-1602.

\bibitem{baez}
J. Baez,
{\em Generalized Measures in Gauge Theory}
Letters in Mathematical Physics 
Volume: 31, Issue: 3 (1993), 12.

    \bibitem{tangent} 
 J. Cari\~{n}ena ,  J. Clemente-Gallardo, E. Follana, J. Gracia-Bond\'{i}a, A. Rivero,
 and J. Varilly,
 {\em Connes’ tangent groupoid and strict quantization}, J. Geom. Phys. 32 (1999), 79–96.

   
  \bibitem{stdtriple}
A. Chamseddine, A. Connes, and M. Marcolli,
{\em Gravity and the standard model with neutrino mixing},
Adv. Theor. Math. Phys. Volume 11, Number 6 (2007), 991-1089.

\bibitem{ncg}
A. Connes,
{\em Noncommutative Geometry}, Academic Press, London, 1994.

\bibitem{lai}
A. Lai, 
{\em 
The JLO Character for The Noncommutative Space of Connections of Aastrup-Grimstrup-Nest},
 arXiv:1010.5226 (2010).
   
   \bibitem{LQG}
   T. Thiemann
   {\em Introduction to Modern Canonical Quantum General Relativity}
   arXiv:gr-qc/0110034 (2001).
     \end{thebibliography}
\end{document}